\documentclass[11pt]{article}
\usepackage{geometry}
\usepackage{makeidx}  
\usepackage{booktabs} 
\usepackage{microtype}
\usepackage{color}
\usepackage[algo2e,ruled,vlined,linesnumbered]{algorithm2e}
\usepackage{colortbl}
\usepackage[nolist]{acronym}
\usepackage{tabularx}
\usepackage{graphicx}
\graphicspath{{./graphics/}}
\usepackage{siunitx}
\usepackage{tikz}
\usetikzlibrary{automata}
\usepackage{url}
\usepackage{booktabs}
\usepackage{rotating}
\usepackage[mathcal]{euscript}
\usepackage{times}
\usepackage{titlesec}
\usepackage{authblk}
\usepackage{todonotes}
\usepackage{tcolorbox}
\usepackage{amsfonts}
\usepackage{amssymb,amsmath,amsthm}
\usepackage[primitives]{cryptocode}

\newtheorem{definition}{Definition}

\newtheorem{lemma}[definition]{Lemma}
\newtheorem{corollary}[definition]{Corollary}

\newtheorem{theorem}[definition]{Theorem}

\newcommand{\R}{\mathbb{R}}

\newcommand{\p}{\mathbb{P}}

\SetCommentSty{mycommfont}

\newenvironment{smartcontract}[1][]
  {
   \begin{algorithm2e}[#1]%
	   \DontPrintSemicolon
  }{\end{algorithm2e}}

\newcommand{\ignore}[1]{\relax} 
\usepackage[mathcal]{euscript}
\usepackage{float}
\usepackage{hyperref}
\definecolor{linkgreen}{RGB}{0,130,0}
\hypersetup{
    colorlinks=true,       
    linkcolor=red,          
    citecolor=linkgreen,        
    urlcolor=linkgreen          
}
\title{\textbf{TRIDEnT: Building Decentralized Incentives for Collaborative Security}}

\author[1]{Nikolaos Alexopoulos}
\author[2]{Emmanouil Vasilomanolakis}
\author[3]{Stephane Le Roux}
\author[1]{Steven Rowe}
\author[1]{Max M\"{u}hlh\"{a}user}

\affil[1]{Technische Universit{\"a}t Darmstadt, Germany}
\affil[2]{Aalborg University, Denmark}
\affil[3]{LSV, ENS Paris-Saclay \& CNRS, Université Paris-Saclay, France}
\date{{\tt \small{\{alexopoulos@tk.tu-darmstadt.de, emv@cmi.aau.dk, leroux@lsv.fr, steven.rowe@web.de, max@informatik.tu-darmstadt.de\}}}}

\begin{document}

\maketitle   
\begin{abstract}
	Sophisticated mass attacks, especially when exploiting zero-day vulnerabilities, have the potential to cause
	destructive damage to organizations and critical infrastructure. To timely detect and contain such attacks,
	collaboration among the defenders is critical. 
	By correlating real-time detection information (alerts) from multiple sources (collaborative intrusion detection), defenders can
	detect attacks and take the appropriate defensive measures in time.
	However, although the technical tools to facilitate collaboration exist, real-world
	adoption of such collaborative security mechanisms is still underwhelming. This is largely due to a lack of trust and participation incentives for companies and organizations.
	This paper proposes \textit{TRIDEnT}, a novel collaborative platform that aims to \emph{enable} and \emph{incentivize} parties
	to exchange network alert data, thus increasing their overall detection capabilities. \textit{TRIDEnT} allows parties that may be in a \textit{competitive} relationship, to selectively \emph{advertise},
	\emph{sell} and \emph{acquire} security alerts in the form of (near) real-time peer-to-peer streams.
	To validate the basic principles behind \textit{TRIDEnT}, we present an intuitive game-theoretic model of alert sharing, that is of independent interest, and show that collaboration is bound to take place infinitely often. Furthermore,
	to demonstrate the feasibility of our approach, we instantiate our design in a decentralized manner using Ethereum smart contracts and provide a fully functional prototype.
\end{abstract}
\emph{``No one can build his security upon the nobleness of another person''}\\
\indent -- Willa Cather, Alexander's Bridge
\section{Introduction}
In recent years, cyber-attacks have grown in impact and have affected millions of people and organizations all over the world. Recent examples, like the infamous Wannacry ransomware attack~\cite{ehrenfeld2017wannacry}
and the Mirai \ac{IoT} botnet \cite{antonakakis2017understanding}, incurred losses calculated to amount to hundreds of millions of US Dollars~\cite{miraicost}.
Other attacks were aimed at acquiring restricted data, like the so-called Red October malware which stole vital information from government and research institutions.
The latter operated undetected for five years, according to Kaspersky estimates~\cite{redOctoberAttack}. Similar attacks are provisioned to become more common and disruptive, since attacker incentives grow as more people depend on interconnected devices
(\ac{IoT}) and critical infrastructure components.
Worryingly, all attacks mentioned above were relatively simple, using known ``old'' vulnerabilities (e.g.\ Wannacry used the infamous NSA-related EternalBlue exploit), yet caused
significant disruption. Mass attacks exploiting yet-unknown (zero-day) vulnerabilities can have orders of magnitude larger attack surface and therefore
be destructive.

To timely detect signs of attacks and take appropriate action, firms and organizations use, among other countermeasures, \acp{IDS}.
\acp{IDS} utilize techniques like signature matching (e.g.\ \emph{Snort}~\cite{roesch1999snort} and \textit{Bro} \cite{paxson1999bro}), or sophisticated anomaly-based detection algorithms for the detection of unknown attacks (zero-days) \cite{Chandola:2009:ADS:1541880.1541882}.
However, given the advancing complexity and severity of attacks, isolated \acp{IDS} that only monitor one part of the network are not adequately effective.
Attackers may have objectives that require a series of steps, including information gathering (e.g.\ mass network scans) and vulnerability assessment that may not be treated as a threat when viewed in isolation.
Consequently, these complex, multi-step intrusion attempts (e.g.\ Advanced Persistent Threats) can only be detected by correlating information from different parts of the network~\cite{cuppens2002alert}.
Transfer of information is also of notable benefit in scenarios where adversaries mount mass attacks targeting a large number of defenders by utilizing similar techniques (e.g. the same zero-day vulnerability). Early warning signs (e.g. in the form of a network trace, a malicious URL or IP address)
from defenders that were targeted by a given attack can greatly benefit others against (similar) future
attacks against them. The latter scenario is particularly important, since such mass attacks are widespread in the wild~\cite{allodi2017work}, and as already noted, can have destructive consequences.
Threat intelligence sharing is therefore critical in defending against modern attacks, as it greatly enhances the defensive capabilities of isolated \acp{IDS}.

The topic of collaborative security through security alert sharing and correlation has attracted research interest in itself, with a variety of proposals on
exchange mechanisms~\cite{vasilomanolakis2015survey,zhou2010survey,meng2015collaborative,lincoln2004privacy}, and standardization efforts~\cite{wood2007intrusion,danyliw2016incident}. 
On the one hand, such proposals allow the exchanged data to retain its utility, and on the other hand protect the source of the alerts by filtering out sensitive information (e.g. IP addresses).
However, there is still great reluctance from companies and organizations to share alert data, as
financial incentives for doing so are not clear. This further highlights the central role of (security) economics~\cite{anderson2001economics,anderson2006economics}
in any real-world security problem. Companies are reluctant to engage in sharing activities, due to the fact that
their competitors can take advantage of
their (financial) investments in the detection effort and the ensuing sharing process,
i.e. competitors may \emph{free-ride}\footnote{Actually free-riding is the Nash equilibrium in many security information sharing models~\cite{laube2017strategic}}. Apart from the basic cost of personnel and hardware/software resources dedicated to maintaining the
sharing infrastructure, information leakage when openly publishing (even anonymized) alerts may lead to additional risks. These can come in the form
of law-infringement violations~\cite{tosh2015game} or leaking information to an attacker concerning the defender's location and
defensive capabilities~\cite{shmatikov2007security}.
Finally, distrust among the participants further aggravates privacy concerns, while also limiting
confidence in the received information.
For these reasons, investment in \acp{CIDS} is viewed cautiously by companies and organizations~\cite{ponemon2018study}.

As recounted above, incentive issues have considerably limited the adoption of such collaborative
systems, even though political initiatives (e.g. by the US federal government and the EU -- see e.g.\ \emph{PROTECTIVE}~\cite{PROTECTIVE18})
highly encourage the practice.
Note that the lack of incentives is detrimental in this case because,
contrary to other systems that work and rely on the selfless behavior of community members (e.g.\ the Tor network),
in the case of security alert sharing the main stakeholders are companies that operate with selfish, rather than altruistic motives.
Existing threat intelligence sharing platforms (e.g.\ MISP~\cite{wagner2016misp}, IBM's ``X-Force Exchange''\footnote{https://exchange.xforce.ibmcloud.com/}, or Facebook's ``ThreatExchange''\footnote{https://developers.facebook.com/programs/threatexchange/})
do not consider incentives and trust, and are intended to be operated by central trusted third parties. The latter is an important limitation, since organizations (or even governments) in any kind of a competitive relationship,
will be uneasy providing valuable data to services fully controlled by their competitors.

Motivated by the aforementioned challenges, we proceed to investigate and address the lack of incentives for security information sharing.
In particular, we introduce \textit{TRIDEnT} \textit{(Trustworthy collaboRative Intrusion DETection)}, a novel alert\footnote{When we refer to alerts, we generally include any information (human or machine-produced) that can be used in threat detection or mitigation.} exchange
platform that aims to \emph{enable} and \emph{incentivize} parties to exchange network alert data.
Specifically, \textit{TRIDEnT} introduces an open, carefully designed alert marketplace that offers the required functionality and incentives for entities to share alert data,
while providing the environment for trust relations between them to develop.
To formally investigate the incentives issues of network alert sharing, we provide a novel
game-theoretic model that acts as a validation for the basic principles of our marketplace design.
Finally, we show the feasibility of our proposal by providing a decentralized prototype on top of the Ethereum smart contracts platform.

We believe that this work is a sizeable step forward towards understanding the economic challenges of deploying collaborative security
mechanisms and providing the required infrastructures to enable collaboration among competitors without requiring trust on a single
service provider.
Our contributions can be summarized by the following points:

\smallskip\par\noindent\textbf{The TRIDEnT platform design:}
\emph{TRIDEnT} is a collaborative platform for alert data sharing that facilitates the creation of P2P channels for collaboration among interested parties.
In \textit{TRIDEnT}, alert producers offer their data in the form of live streams, selectively to parties of their choice, thus creating a sharing overlay based on their trust relations.
To attract interested collaborators, they advertise their streams by including information about the data, in the form of tags.
For example, producers may advertise tags
relating to the type of attacks (e.g.\ \ac{DDoS}, malware), the detector (e.g.\ \ac{IDS}, honeypot), the type of network (e.g.\ backbone, corporate),
and so on. 
Interested parties can buy streams with some form of currency\footnote{For the rest of the paper we refer to the exchange medium as the \emph{TRIDEnT token}, however any currency can be used, i.e. a dedicated token is not technically necessary with our existing
incentives analysis. However, we do not rule out that a dedicated token with regulated
supply can be beneficial in other game-theoretic models.}
to pay the producers and start streaming in real time. A rating and trust management system is included to ensure that participants offer good quality
alert data.

\smallskip\par\noindent\textbf{Game-theoretic model of alert exchange:}
To validate the core principles behind \emph{TRIDEnT}, we introduce an intuitive game-theoretic model of network alert sharing,
where two rational selfish defenders try to optimize their payoff against an attacker who performs attacks stochastically. We pinpoint a pattern
of attacks that leads to the exchange of information and we show that this pattern repeats infinitely often, and thus collaboration
is bound to take place. Our model differs from other proposed models for security collaboration (e.g.~\cite{zhu2012guidex} and~\cite{jin2017tradeoff}) in that
it assumes a competitive relation between selfish players that are only interested in sharing information if it will minimize their
expected costs. Thus, security collaboration is not viewed as a matter of good will or legislative pressure, rather as a self-serving
action. Furthermore, our model differs from other strategic information sharing games (e.g. ~\cite{laube2017strategic}), in that the act of sharing information is not considered to
be ``for free'' (as processing and privacy leaks incur costs).

\smallskip\par\noindent\textbf{Instantiation on Ethereum:}
We showcase the feasibility of our system in an open, decentralized setting (anyone can participate), with an implementation on the \emph{Ethereum} platform.
For readers unfamiliar with the basic functionality of Ethereum, a brief
overview is given in Section~\ref{sec:background}.
We describe the design along with the basic components of the prototype 
and show that transaction fees incurred by the system are projected to remain negligible.
Note that the same design can straightforwardly be
adapted to a permissioned setting, e.g.\ by implementing it on \emph{Hyperledger Fabric}~\cite{cachin2016architecture}, in the case of a closed group of known, yet selfish and competitive organizations.

The remainder of this paper is structured as follows. Section~\ref{sec:background} offers some background information, required for the comprehension of the paper. Section~\ref{sec:market} presents the collaborative platform and focuses on its architecture and operations.
Afterwards, in Section~\ref{sec:model} we validate the core ideas of \textit{TRIDEnT} in a game-theoretical model of independent interest. Then, Section \ref{sec:prototype} presents a prototype implementation of \textit{TRIDEnT} and its evaluation. Section~\ref{sec:related} goes over the related work, Section~\ref{sec:limits} discusses limitation of our work, while Section~\ref{sec:concl} concludes the paper.

\section{Background}\label{sec:background}
In this section we provide some basic background about \textit{i)} collaborative security and \textit{ii)} \textit{Ethereum} and smart contracts.

\medskip\par\noindent\textbf{Collaborative Security and \acp{CIDS}}
The term collaborative security is an encapsulation of the idea that \textit{combining knowledge from different sources can be of benefit for the detection and mitigation of attacks}. The reasoning behind this argument is well studied \cite{zhou2010survey}, and can be summarized in the following: detectors, e.g., anomaly detection algorithms, can be improved by enhancing the input data, alert correlation is more effective when the data volume increases, a number of attacks, e.g., malware spreading, can be contained even before seen locally when they are anticipated as a result of collaboration.
For instance, B\"ock et al.~\cite{leon2018next} recently showed that correlation of data from different sensors is necessary to enumerate advanced P2P botnets.
\acfp{CIDS} further formulate the aforesaid idea by implementing it in the form of a system \cite{vasilomanolakis2015survey}. Besides academic work, e.g., \cite{vasilomanolakis2015skipmon,fung2011dirichlet}, there have been various attempts to realize \acp{CIDS} \cite{ullrich2000dshield,PROTECTIVE18}. Nevertheless, the majority of the proposed systems are either theoretical or assume that participants are trusted and are willing to collaborate. 

\smallskip\par\noindent\textbf{Ethereum and smart contracts}
\emph{Ethereum}~\cite{wood2014ethereum} is a decentralized platform (distributed ledger) for Turing-complete applications, facilitated by smart contracts,
with a native cryptocurrency called \emph{ether}. The state of the platform is collectively maintained by its users via a consensus mechanism, thus guaranteeing its
correct operation under reasonable security assumptions (honest majority of computing power and relatively good network connectivity). Thus, the platform is as if it were operated by a trusted party with public internal state. Smart contracts can be written in high level languages (e.g., \emph{Solidity}),
compiled and executed by the \emph{\ac{EVM}} which can be thought of as a large decentralized computer,
spread across the world. A smart contract can be deployed by issuing a ``deploy'' transaction
to the network, and functions of the smart contract can be called by referencing the address where the contract is stored
in the ledger. Transactions are statements issued by a party, signed by the party's secret key $s_k$, and sent to the
network. They can either invoke functions in a contract, transfer ether between accounts, or deploy new contracts.
Because data stored in the ledger is replicated among all participants, and code affecting this data has to
be executed by all participants, each transaction incurs a \emph{gas} cost to the issuing party, depending
on the ``work'' the participants have to invest in processing it (e.g., the computational complexity of the operations).
Each operation in the \ac{EVM}'s instruction set is associated with a specific gas cost, and the issuing party has to include
a corresponding offer of ether per gas unit consumed, in order to incentivize the network to run the computations
requested and update the global state.

\section{The TRIDEnT sharing platform}\label{sec:market}
We begin the presentation of our platform by familiarizing the reader
with its basic components and their interaction, along with a
high-level overview of their functionality. We then proceed to describe
in more detail the more interesting parts and design choices. An avid reader
can infer all details about our platform from the smart contract pseudocode of our
implementation (see Section~\ref{sec:prototype}).
\begin{figure}
	\begin{center}
		\includegraphics[scale=0.7]{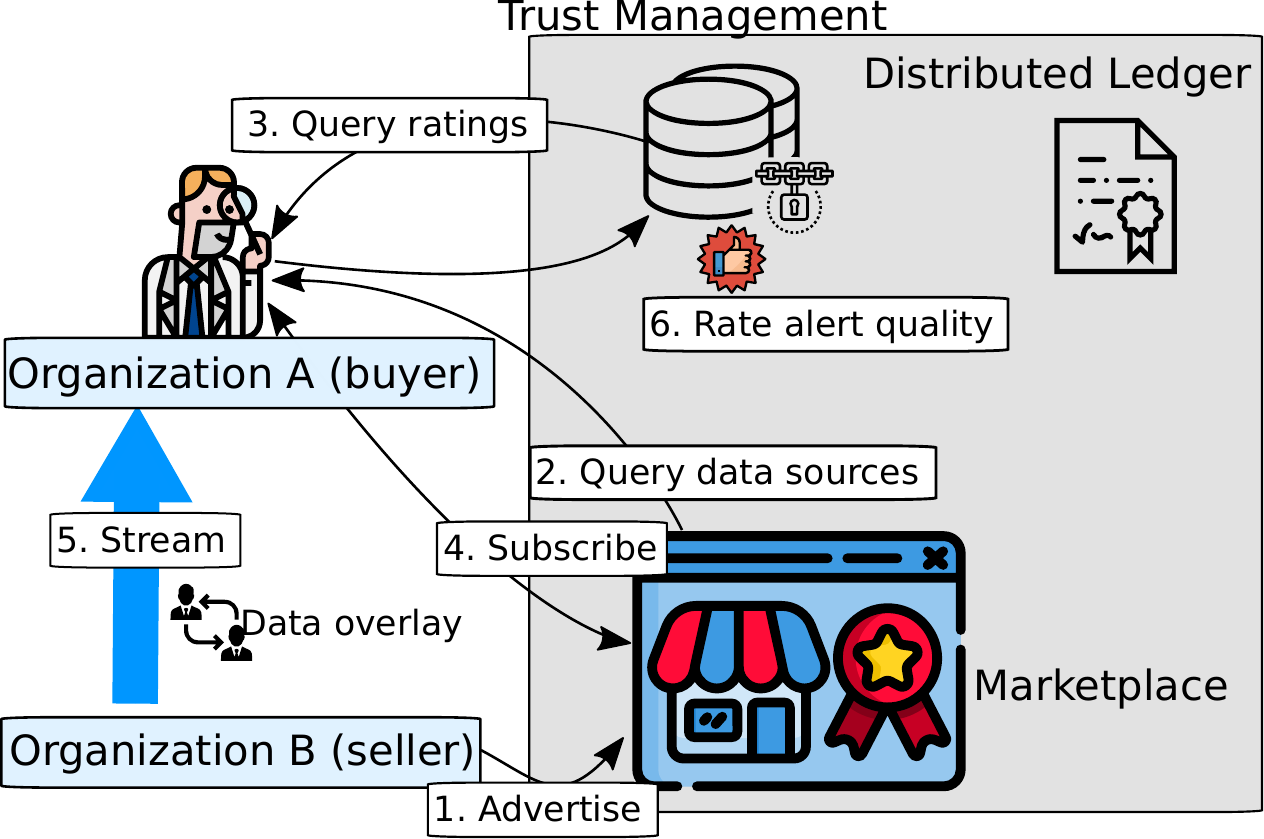}
		\caption{\textit{TRIDEnT} workflow.}
		\label{fig:arch_pic}
	\end{center}
\end{figure}
\subsection{System architecture and overview}
\textit{TRIDEnT} can be conceptually dissected into the following layers:
\smallskip
\par\noindent\textbf{Distributed Ledger Layer: }
The system is built on top of a distributed ledger (DL), e.g. the Ethereum platform~\cite{wood2014ethereum}, which offers strong consistency
and availability guarantees and allows participants to execute arbitrary (Turing-complete) operations
on its state. This is the base layer of the framework and lays the trust foundations for the following
layers, emulating a trusted third party. The benefits of using a DL as a sharing
infrastructure have been documented recently~\cite{alexopoulos2017towards,webstersharing}.\smallskip
\par\noindent\textbf{Trust Management Layer: }
As participants can behave maliciously at times and towards certain parties,
a trust management mechanism (i.e. generalized reputation system) is necessary to support good behavior. Ratings
from buyers follow transactions (stream establishment),
and are stored on the DL. At each time, a peer is able to
calculate a local trust score for each other peer, via a trust calculation algorithm of
her choice. In our design (see Section~\ref{sec:trust}) we use an adaptation of the
Bayesian mechanism of~\cite{ries2007certain}.
\smallskip
\par\noindent\textbf{Marketplace Layer: }
The system provides economic incentives for honest behavior by utilizing
a token-based economic mechanism. Tokens can be thought of as a currency
special to the system, used to buy streams from other participants. 
In the case where a native cryptocurrency is offered by an underlying
distributed ledger, then this can be used as a token.
A game-theoretic justification for the alert marketplace can be found in Section~\ref{sec:model}.
\smallskip
\par\noindent\textbf{Data Overlay Layer: }
After establishing a connection, a seller and a buyer open a channel where
alert data and tokens are exchanged. Multiple such channels create sharing overlays.
We employ off-chain transactions (see the Raiden network\footnote{https://raiden.network/}) to enable seamless and
near-instantaneous token transfers.
\medskip\\
An intuitive description of the basic functionalities of the platform,
along with the connections among the layers can be seen in Figure ~\ref{fig:arch_pic}.
A simplified example operation scenario follows.
Organization B, acting as an alert \emph{seller} advertises the information he can offer
(e.g. alerts from his industrial IDS) on the marketplace.
Organization A, acting as an alert \emph{buyer}
(i.e. party interested in acquiring alert data) queries the marketplace where
advertisements of alert streams are posted.
The buyer also computes trust
scores for each seller by consulting previous ratings, and decides whether she wants to make a subscription offer matching a stream's
requested price. Assuming the buyer makes an offer for a stream, and the corresponding seller decides to accept it (the seller can also
consult the buyer's trust score as risk in these transactions is bidirectional), a data stream and a payment channel are established.
The buyer may now rate the stream (and indirectly the seller) based on its perceived quality and whether or not it matches the
advertisement. The rating is stored on the ledger. Details about the platform's operation follow.

\subsection{Marketplace functionality}
The core functionality of the platform is summarized by the marketplace function definitions of Table~\ref{tab:functions}, which are
to be instantiated by a smart contract.
Our actual instantiation of these functions is reserved for Section~\ref{sec:prototype}.
Users register on the platform by providing their public key and ``burning'' a reasonable amount of \textit{ether} to create an initial trust value ({\tt{register}}). Then they can advertise streams
({\tt{advertise}}), make offers ({\tt{mkOffer}}) for existing advertisements, or accept offers ({\tt{accOffer}}) for their own advertisements.
In particular, when making an offer, a security deposit has to be provided in order to incentivize the act of providing ratings (see below).
The marketplace also offers functionality for deleting offers ({\tt{delOffer}}) and
unsubscribing from streams ({\tt{unsubscribe}}), while enforcing the necessary constraints. Finally, ratings can be provided via
calling the {\tt{rate}} function. A deposit is used to incentivize buyers to provide ratings for streams.
Note that our design is generic and can in theory provide incentives and trust for any kind of information sharing platform.
However, security information sharing is particularly challenging, as it will become obvious in the following paragraphs,
and we therefore focus on it.
In the remainder of this section we present some notable characteristics of \textit{TRIDEnT}'s operation.
\begin{table}[h]
{\footnotesize
\centering
\begin{tabular}{lll}\toprule
\textbf{Function} & \textbf{Description} & \textbf{Constraints} \\ \midrule
\tt{register} &  used for initial registration & \\
&  burns ether & \\ \midrule
\tt{advertise} &  create new advert. with & \\
 & chosen tags & \\ \midrule
\tt{rmAdvert} &  remove given advert.  & \\
 & + related offers and subscr. & \\ \midrule
 \tt{mkOffer} &  create offer for given advert. &  {\footnotesize{deposit has to be provided}}\\ \midrule
 \tt{delOffer} &  delete given offer &  {\footnotesize{caller has to be the advert.}}\\
 & & {\footnotesize{publisher or offer creator}} \\\midrule
 \tt{accOffer} &  delete offer and create subscr. &  {\footnotesize{caller has to be the advert.}}\\
 & & {\footnotesize{publisher}} \\ \midrule
 \tt{unsubscribe} &  delete subscription & {\footnotesize{caller has to be the advert.}}\\
 & & {\footnotesize{publisher or the subscriber}} \\ \midrule
 \tt{rate} & add rating & {\footnotesize{only one rating/subscr.}}\\
 &  & {\footnotesize{caller has to be subscriber}} \\
 &  & {\footnotesize{timer must not be expired}} \\ \midrule
\end{tabular}
\caption{Marketplace function definitions with constraints.}
\label{tab:functions}
}
\end{table}
%
\smallskip\par\noindent\textbf{Alert advertisement format: }
An alert advertisement includes three important characteristics of the stream, namely the expected mean
throughput of the stream (e.g. alerts per hour), the price of streaming (e.g. tokens per alert batch), and a
list of strings (\textit{tags}) offering information on the type and origins of the alerts in the stream. A simple example is shown in Table~\ref{tab:advert}.
Note that the formats presented in this section serve as simple abstract templates to showcase the functionality of \emph{TRIDEnT}, meaning that
they can be readily enhanced as required.

With regard to the tags, we have classified this information into three categories: (i) type of the detector (e.g. honeypot, \ac{IDS}, firewall, etc.), (ii) type
of the network (e.g., university network, home network, backbone), and (iii) type of observed attacks (e.g., \ac{DDoS} attacks, malware, botnet,
APTs, port scans, etc.). Here we note that the third type of information is valid only for streams that provide historical data
(not live data). Whereas in first inspection, it may seem useless to enable exchange of stale alerts in our system, this is
often useful. For example, companies specializing in developing solutions for detecting and anticipating specific kinds of attacks would be very interested in alerts that have been labeled accordingly.
\smallskip\par\noindent\textbf{Advertising and subscribing to streams. }
Participants of the CIDS can subscribe to alert streams of other parties and also advertise their own streams on the ledger.
Participants are typically \acp{IDS} or honeypots but also other \acp{CIDS} or humans, e.g., security experts.
For the CIDS as a whole, it is desirable that parties who consume alert data from the other parties, also provide their own data.
An advertisement contains the id of the publisher that can be used to calculate her trust value,
and tags that describe which kind of alerts the subscriber can expect, as explained above.
This allows the node to browse all available alert streams.
{\small
\begin{table}[h]
\centering
\begin{tabular}{l}
\begin{tcolorbox}
{\tt Publisher:} GoodIDS\\
{\tt Expected throughput:} 10/hour\\
{\tt Price:} 1\$/alert\\
{\tt Detector type:} IDS\\
{\tt Type of network:} industrial\\
{\tt Type of attacks:} DDoS, botnet\\
\end{tcolorbox}
\end{tabular}
\caption{Example advertisement.}
\label{tab:advert}
\end{table}}
%
%
\smallskip\par\noindent\textbf{Stream establishment and operation.}
Figure~\ref{pr:stream} presents the stream establishment protocol of \emph{TRIDEnT}, involving a Buyer, a Seller and the
system's smart contract, which facilitates operations. First, involved parties create and register a (public) key
to be used for authentication and encryption. In practice, an implementation of OpenPGP (e.g. GnuPG) could be used for all cryptographic operations
in this protocol, potentially with different keys for authentication and encryption, as per standard practice.
Upon deciding to make an offer for a stream corresponding to an advertisement $S_i$, the Buyer calls the \texttt{mkOffer} function of the smart contract,
thus transmitting the offer to the DL.
Upon deciding to accept the offer, the seller publishes on the ledger the signed and encrypted address of the socket (IP address, port) where the buyer can connect to
start streaming. The address is encrypted with the buyer's public key (initialized upon registration) and ensuing connections to the socket (assuming initial authentication
is successful) are encrypted and authenticated. Execution paths concerning decisions not to accept the offer are pruned from Figure~\ref{pr:stream} for brevity.
\begin{figure}
		\begin{center}
	\fbox{
		\procedure{Stream establishment protocol ($\mathcal{SE}$)}{
			\textbf{Buyer} \< \hspace{0.3cm}\textbf{Contract} \< \hspace{1.7cm} \textbf{Seller} \pclb
		\pcintertext[dotted]{initialization}
		\< \< (sk_S, pk_S) \gets Gen(\secparam) \\
		\< \< \sendmessageleftx[2cm]{0}{\text{\texttt{register}}(pk_S)}  \\
		(sk_B, pk_B) \gets Gen(\secparam) \< \< \< \< \\
		\sendmessageright*[2cm]{\text{\texttt{register}}(pk_B)} \< \< \pclb
		\pcintertext[dotted]{exchange}
		\sendmessageright*[2cm]{off=\text{\texttt{mkOffer}}(S_i)} \< \\
		\< \< m \gets \sig_{sk_S}(\text{ip:port})\\
		\< \< c \gets \enc_{pk_S}(m)\\
		\< \< \sendmessageleftx[2cm]{0}{\text{\texttt{accOffer}}(off, c)}  \\
		\sendmessageright*[2cm]{\text{read}(c)} \< \\
		\text{ip:port} \gets \dec(c) \< \< \\
		\>  \sendmessagerightx[4cm]{4}{connect(\text{ip:port}) } \> \\
		\> \sendmessageleftx[4cm]{4}{nonce} \> \\
		ss \gets \sig_{sk_B}(nonce) \< \< \\
		\>  \sendmessagerightx[4cm]{4}{ss} \> \pclb
		\pcintertext[dotted]{success} \\
		\< \< \verify(ss, nonce, pk_B) \stackrel{?}{=} 1 \\
		\> \sendmessageleftx[4cm]{4}{\text{accept and start streaming}} \> \pclb
		\pcintertext[dotted]{failure}
		\< \< \verify(ss, nonce, pk_B) \stackrel{?}{=} 0 \\
		\> \sendmessageleftx[4cm]{8}{\text{reject and drop}} \>
	}
	}
	\caption{Stream channel establishment via a smart contract. Notation follows standard practice, with $Gen$ the key generation
	operation (with security parameter \secparam), $\sig$ the digital signature operation, $\enc$ and $\dec$ the operations of encryption and decryption, and $\verify$
	signature verification. Typewriter font is used for calls to the marketplace functions.}
	\label{pr:stream}
\end{center}
	\end{figure}

\emph{TRIDEnT}, as a basic infrastructure, can be used to share any kind of
security-related information (malware signatures, software vulnerabilities, etc.)
However, in this paper we focus on network alerts, generated by \acp{IDS} such as Snort or anomaly detection systems,
given in a suitable format for exchange that can be readily processed.
We propose the state-of-the-art STIX format~\cite{stixformat}, although other formats could also be used. For a comparison of alert exchange formats see~\cite{menges2018comp}.
A general template is given in Table~\ref{tab:alert}.
A specific STIX alert in JSON format signaling a malicious URL follows in Table~\ref{tab:json}.
We refer the reader to the documentation of STIX\footnote{https://stixproject.github.io/about/} for a more detailed description.
\emph{TRIDEnT} is agnostic to the specific format used.
{\small
\begin{table}[h]
\centering
\begin{tabular}{l}
\begin{tcolorbox}
{\tt Time:} creation/sending time\\
{\tt Source:} origin of attack\\
{\tt Target:} target of attack\\
{\tt Classification:} name and/or CVE\\
{\tt Assessment:} e.g. severity, potential impact, etc.\\
\end{tcolorbox}
\end{tabular}
\caption{Generic alert format.}
\label{tab:alert}
\end{table}}
{\small{
\begin{table}[h]
\centering
\begin{tabular}{l}
\begin{tcolorbox}
\{\\
      ``type'': ``indicator'',\\
      ``id'': ``indicator--9299f726-ce06-492e-8472-2b52ccb53191'',\\
      ``created\_by\_ref'': ``identity--39012926-a052-44c4-ae48-\ldots'',\\
      ``created'': ``2017-02-27T13:57:10.515Z'',\\
      ``modified'': ``2017-02-27T13:57:10.515Z'',\\
      ``name'': ``Malicious URL'',\\
      ``description'': ``This URL is potentially associated with malicious
      activity and is listed on several blacklist sites.'',\\
      ``pattern'': ``[url:value = 'http://paypa1.banking.com']'',\\
      ``valid\_from'': ``2015-06-29T09:10:15.915Z'',\\
      ``labels'': [``malicious-activity'']\\\}
\end{tcolorbox}
\end{tabular}
\caption{Example STIX alert (JSON format). More examples can be found in the docs\protect\footnotemark.}
\label{tab:json}
\end{table}}}
\footnotetext{https://oasis-open.github.io/cti-documentation/stix/examples}

In order to support real-time micro-payments, we employ off-chain payment channels, as provided by the Raiden Network.
Specifically we utilize the $\mu$-Raiden variant, which supports only one receiver but is simpler and cheaper in terms of transaction costs.
A uni-directional payment channel is created upon successful establishment of a stream between the buyer and the seller. Subsequent payments
take place via the channel, thus allowing per-alert payments with practically zero transaction fees.
\subsection{Building trust}\label{sec:trust}
As recounted earlier, trust among the participants of collaborative security mechanisms is of paramount importance.
The decentralized smart-contract-based design of \emph{TRIDEnT} does not require a trusted party to act as the operator
of the sharing service, thus avoiding the need for the participants' reliance on a common provider.
However, some parties might try to profit from the streams of the other parties while themselves providing fake or bad quality data.
There is also a risk of malicious participants that may try to destroy the \ac{CIDS}, e.g., by providing wrong alert data or by bad-mouthing
honest parties.
Therefore, some form of trust management, such as a reputation system, is required in order to
protect participants from malicious insiders that disseminate low quality, malformed,
or misleading alerts. \textit{TRIDEnT} comes with a built-in trust bootstrapping and rating system, thus
providing the foundations for trust assessment.
\smallskip\par\noindent\textbf{Trust bootstrapping: }Since new parties may enter the system at any time (\emph{TRIDEnT} is
meant to be an open system), a way to bootstrap their standing in the community is needed, even if they do not have real-world social
ties to other participants. At the same time, this bootstrapping mechanism should mitigate sybil attacks.
In \emph{TRIDEnT}, trust bootstrapping is achieved via Proof-of-burn.
To register in the system, a participant ``burns'' an amount of cryptocurrency in order to prove her commitment to the community.
Cryptocurrency ``burning'' refers to the act of rendering an amount of currency provably un-spendable, equivalent to actual burning of currency notes.
The amount required to achieve a baseline trust value and start interacting in the \emph{TRIDEnT} community should remain low in order not to be considered an
entry barrier for new participants. On the other hand, the amount required to achieve considerably higher initial trust should increase quickly, as to
make it difficult for malicious participants to acquire high initial trust and harm the system. The exact baseline amount is a deployment decision,
not addressed in this paper.
\smallskip\par\noindent\textbf{Rating: }
Transaction ratings are a widely used and empirically effective tool for building trust in marketplace environments. Rating in \emph{TRIDEnT} is intertwined
with the other marketplace operations. After a stream is established, the buyer can rate the seller by
calling the \texttt{rate} function of the smart contract. This rating is thus part of the global state of the underlying \ac{DL} and therefore
visible to all participants. After the stream is closed, the buyer can finalize her rating by calling the \texttt{rate} function once again. This
is to disallow sellers to provide good quality data up until they are rated positively, and then decrease quality. The perceived quality of a
stream depends on whether data provided are judged to be real, useful, and corresponding to its advertisement.

A major problem faced by rating systems is user unwillingness/laziness to provide ratings, especially in the case where it involves extra
effort or possible additional transaction fees. In \emph{TRIDEnT}, users are monetarily incentivized to provide ratings by mechanisms built in the
system. Specifically, upon making an offer, the buyer deposits a small but not negligible fixed fee to the marketplace smart contract.
This fee is returned to the buyer if the offer is rejected or, in two rounds, upon submitting a rating. It is easy to see that if the
fixed fee is considerably larger than the transaction fee required to call the \texttt{rate} function, buyers are incentivized to
follow up any purchase with a rating.

The mechanism described above and implemented in our system, offers clear incentives for buyers to submit ratings. However, incentives for providing \emph{correct}
ratings may also be needed. Although these fall outside of the scope of this paper, \emph{TRIDEnT} can straightforwardly support algorithms
proposed in literature to overcome this problem, e.g. the elegant idea proposed by Jurca and
Faltings~\cite{jurca2003incentive}. The idea is to establish a side-payment channel
organized by a set of broker agents that buy and sell ratings.
They reward (pay) an agent who submitted a rating if the following rating
(coming from another agent) agrees with hers. Under well-defined
assumptions (the next experience is more likely to be the same as the previous: $Pr[C^i_{t+1}|C^j_t]>0.5$ for $i=j$),
truthful reporting is a Nash equilibrium. In \emph{TRIDEnT}, no broker agents are required;
the mechanism can be added to the existing protocol and executed by smart contracts. In this sense, \emph{TRIDEnT} is an extensible platform.
\smallskip\par\noindent\textbf{Local trust computation: }The local trust computation algorithm employed by each party to decide whether or not to engage
in exchange activities with another party is subject to choice, and technically not part of the core \emph{TRIDEnT} infrastructure. Generally, trust decisions in this setting are more complex
than just computing a reputation score and require human engagement. However, we still consider trust scores as valuable decision-supporting tools, and in our
instantiation we use an adapted version of \emph{CertainTrust (CT)}~\cite{ries2007certain} to calculate a trust score for each participant.
In contrast to simplistic ``average rating'' trust representations, CT, since it is based on Bayesian statistics principles,
can model the perceived accuracy of the derived trust value, as well as incorporate prior information (corresponding to Bayesian priors).
More details about our construction follow.

\par{(Bayesian evidence-based trust assessment:)}
In the Bayesian representation of \emph{CertainTrust / CertainLogic}~\cite{ries2007certain,ries2011certainlogic,hauke2015statistics},
\emph{trust} is calculated as the expectation/probability:
\begin{equation}\label{eq:trust}
	E = c_e \cdot t + (1-c_e) \cdot f
\end{equation}
where $t$ is the point estimate of the (probability of success) parameter of a binomial distribution, $f$ is the a priori expectation
of the same parameter, and $c_e$ is a factor that ``fades'' the effect of the prior value, as more pieces of
evidence are taken into account, in accordance with Bayesian statistical inference. Evidence in this model
are either positive or negative experiences (binary ratings). The number of positive and negative experiences is denoted by $r$ and
$s$ respectively, and $n=r+s$. The mapping between the evidence space ($r,s$) and the Bayesian probability space is given by:
\begin{equation}
	t = \left\{
		\begin{array}{lll}
			\frac{r}{r+s} & \mbox{if } r+s>0 \\
			0 & \mbox{else}
		\end{array}
	\right.
\end{equation}
\begin{equation}\label{eq:c}
	c_e = \left\{
		\begin{array}{lll}
			0 & \mbox{if } n=0 \\
			\frac{N \cdot n}{2\cdot w \cdot (N-n)+ N \cdot n} & \mbox{if } 0<n<N \\
			1 & \mbox{if } n \geq N
		\end{array}
	\right.
\end{equation}
where $w$ is a normalizing value we set to $1$, and $N$ is a user-set threshold number of evidence needed to achieve
a desired level of significance of the estimate (in the sense of a confidence interval -- see Appendix~\ref{app:setN}
for details on how to compute a mathematically sound N). It is easy to see from Eq.~(\ref{eq:c}) that the value of $c_e$
increases with increasing number of evidence $n$, as intended.

\par{(A-priori trust through proof-of-burn:)}
We consider the PoB-derived baseline trust as the Bayesian \emph{a priori}
expectation $f$ regarding the trustworthiness of an entity.
\begin{figure}
\centering
\includegraphics[scale=0.5]{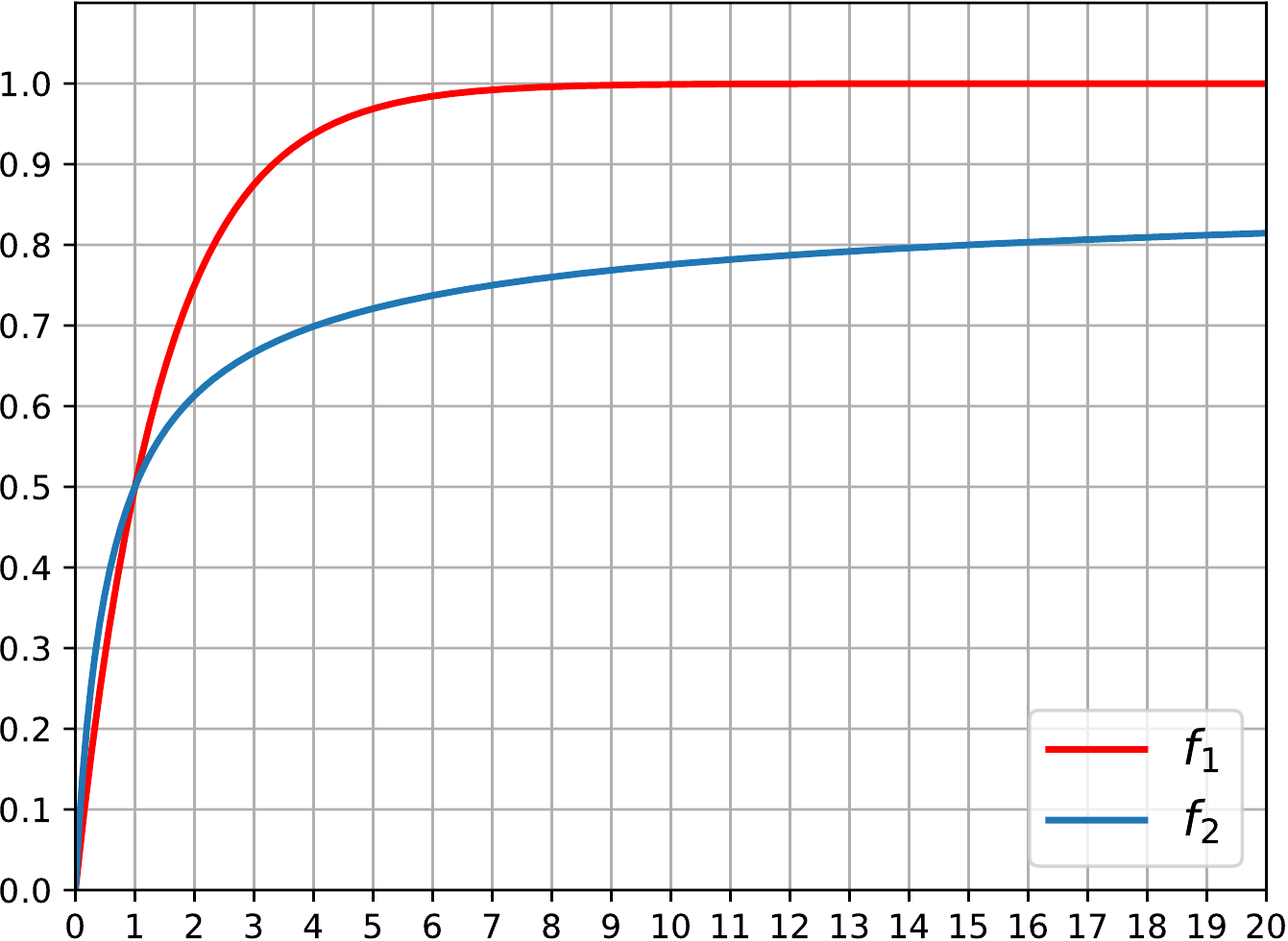}
\caption{Comparison of proof-of-burn candidate functions.}
\label{fig:pob}
\end{figure}
The function $f^*: r \in \mathbb{R}^+ \rightarrow f \in [0,1]$, relating the amount of cryptocurrency units burned, and the resulting baseline trust,
where $r=\frac{x}{c}$ the fraction of the amount of cryptocurrency burned and $c$ the baseline value for achieving initial trust of $0.5$, should
satisfy the following properties:
\begin{equation}
	f^*(0)=0,\ f^*(1)=\frac{1}{2},\ \lim_{r\to\infty} f^*(r) = 1,\ f^{*'}(r)>0,\ f^{*''}(r)<0
	\label{eq:props}
\end{equation}
\ignore{
In~\cite{zindros2014_pseudonymous}, a function $f_1=1-(\frac{1}{2})^r$ is proposed. It is easy to show that the function satisfies the properties above.
However, we observed that even though $f_1''(r)<0$ the function grows too quickly for our use-case. Taking an example, for the default value $c$,
$r=1$, and $f_1(1)=\frac{1}{2}$. The amount $c$ should be set at a reasonable price to allow all willing participants to bootstrap their trust score
in the network. However, assume a party invest four time the default amount, then her a priori trust would be $f_1(4)=\frac{15}{16}$, which is
very close to $1$. Therefore, we need a different function $f_2$ that has the special characteristic that it grows fast until $f_2=\frac{1}{2}$,
and then much slower, making it very difficult to purchase very high bootstrap values.
}
After considering different functions, we arrive at the definition of $f^* \equiv f_1$, with:
\begin{equation}
	f_1(r)=1-\frac{1}{1+log_2(r+1)}
\end{equation}
The definition of $f_1$ satisfies the properties of Eq.~(\ref{eq:props}), and additionally, its growth rate is suitable to our needs. In Figure~\ref{fig:pob},
both $f_1$ and $f_2 =1-(\frac{1}{2})^r$ (proposed in~\cite{zindros2014_pseudonymous}), are depicted. We see that $f_1$ rises faster up to the value of $0.5$, and then its growth rate drops significantly, especially
with respect to $f_2$. By adopting $f_1$, users will be able to acquire a baseline trust value of $0.5$ by spending a reasonable fee, but to acquire
significantly larger initial trust, they will have to pay exponentially more. Even by paying 20 times the default fee, their initial trust
would not be much higher than $0.8$.
\subsection{Privacy discussion}
\textit{TRIDEnT}, by design, offers increased privacy assurances for participants by allowing the formation of P2P
overlays based on trust, rather than indiscreetly publishing alert data.
Alerts are assumed to be sanitized according to best practices (e.g.,\cite{lincoln2004privacy}), limiting the
exposure of private information to collaborators, and the cost of sharing privacy-sensitive information is taken
into account in the game-theoretic analysis of the following section.
Furthermore, alert sanitization could be executed depending on the
trust relation between the collaborators, allowing for less sanitized -- more valuable alerts to flow, as trust
between collaborators grows.
Techniques offering anonymous payments and privacy-preserving reputation (e.g.,~\cite{schaub2016trustless})
could prove beneficial to our design, but are out of our current scope of research. Overall, our system offers
a solid foundation for data exchange, able to support existing and upcoming privacy-enhancing techniques based
on alert sanitization, bloom filters etc., which are not the focus of this paper.
\subsection{Attacks and defenses}\label{sec:attacks}
Apart from the usual attacks and corresponding defenses against trust\,/\,reputation
systems in the marketplace context (see e.g.~\cite{hoffman2009survey}), there are some
specific notable issues that \textit{TRIDEnT} is faced with, that may not be obvious on
first sight.
\smallskip\par\noindent\textbf{Target bad-mouthing attack: }In this scenario, the
adversary attempts to artificially lower the reputation of the target party, so that he
can discredit her reporting competence before attacking her. This way, the adversary
would be able to lower the chance of the attack evidence being spread among the defenders.
The most common way of performing such an attack is creating multiple identities and
giving a lot of negative ratings. Since joining the system comes at a non-negligible cost
and submitting a rating requires starting a stream, which incurs respectable costs,
the attack is expensive for an adversary to perform, and not scalable. Remember, that the
principal aim of \emph{TRIDEnT} is to mitigate mass attacks, not targeted ones.
\smallskip\par\noindent\textbf{Stream reselling: }In this attack, an adversary resells
alerts he has bought from sellers, to other parties. Apart from this being a problem for
the incentives mechanism (the reseller can practically free-ride), this attack can also
compromise the security of the system, as alerts may contain sensitive information and
are intended for use by the designated recipient. Enforcing watermarking on security
alerts seems almost impossible, as the distinctive information would be easy to locate
and remove. A solution with the use of fake alerts as watermarking would also reduce the
value of the alerts as a whole, which goes against the goals of the system. Therefore,
defense against reselling attacks relies on the same mechanism as defense against
low-quality alerts, i.e. the trust management system already described in the previous sections.
Hence, sellers also face risk and should assess their trust in their buyers. That being
said, studying the reselling attack in isolation is a challenging open problem for the
community. For example, mechanisms that would allow honest buyers to discern resellers
could be valuable for the system. One could envision a solution where honest participants
continuously and randomly check a subset of a node's incoming and outgoing streams to
detect resells. This check could be performed in a privacy-preserving manner via
secure multiparty computation. These solutions could of course find applications in
other data sharing applications.

Although a very challenging attack to fully rule out,
reselling also comes with some boundaries that practically reduce its
effect. Most notably, a seller offering a huge amount of streams (due to reselling), would
be very suspicious, and thus would not be preferred by buyers.

\smallskip\par\noindent
After presenting the basic functionality of \textit{TRIDEnT}, we proceed to validate the basic principles behind it with a formal probabilistic model.

\section{A probabilistic and game-theoretic approach}\label{sec:model}
This section presents a simplified mathematical model of our ecosystem. The purpose is not to describe the ecosystem in a comprehensive manner; rather we present a model where several relevant phenomena still occur, including regular exchange of information. 

Our model has five abstract parameters: low and high probabilities to be attacked, cost of sharing information, cost of preparing to defend, and cost of being attacked unprepared. Each of these abstract parameters corresponds to several real-world parameters. For example, the cost of sharing information is the aggregation of writing a report that is meant to be shared, accessing the sharing platform, taking the risk of information leakage, etc.  Further details regarding these parameters can be found in Section~\ref{sect:model-discussion}.

Our model is thus easy to grasp and mathematically tractable, while representing our design and the real world rather faithfully (although not completely). Slight modifications could make our model more faithful but also more complex without necessarily providing new insights about the relevant phenomena. In this context, our current model offers a fair evaluation and justification for our design.

Section~\ref{sect:game-model} presents our model; Section~\ref{sect:game-result} presents our result stating that exchange of information occurs regularly; and Section~\ref{sect:model-discussion} discusses how realistic our model is.

\subsection{The model}\label{sect:game-model}
\par\noindent{\textbf{Overview.}}
Let us define a game between two players (defenders) and an attacker. The game consists of infinitely many rounds. 
At the end of each round each player may or may not be attacked. A player can only observe whether herself is attacked, but at the beginning of each round (beside the first round) she can pay the other player and learn whether the other player was attacked during the previous round. At the beginning of the first round, the players set their respective selling prices for disclosing an attack. In the middle of each round, each player may or may not prepare to defend against an attack.
\par\noindent{\textbf{The attacks.}} In our model, if someone (resp. no one) was attacked during the previous round, each player is attacked independently with probability $q$ (resp. $p < q$). By convention, at the first round the probability $q$ is used instead of $p$. 
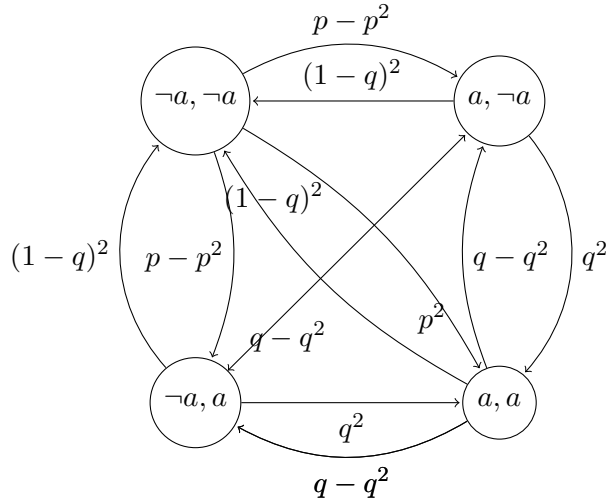
\begin{figure}[h]
	\centering
	\begin{tikzpicture}[shorten >=1pt,node distance=4cm, auto]
	\node[state] (nn) {$\neg a,\neg a$};
	\node[state] (an) [right of = nn] {$a,\neg a$};
	\node[state] (na) [below of = nn]{$\neg a, a$};
	\node[state] (aa) [below of = an]{$a, a$};
	
	\path[->] (nn) edge [bend left] node [above] {$p-p^2$} (an)
	(an) edge node [above] {$(1-q)^2$} (nn)
	(an) edge [bend left = 40] node [right] {$q^2$} (aa)
	(aa) edge  [bend left = 20] node [right] {$q-q^2$} (an)
	(nn) edge [bend left = 20] node [left] {$p-p^2$} (na)
	(na) edge  [bend left = 40] node [left] {$(1-q)^2$} (nn)
	(na) edge node [below] {$q^2$} (aa)
	(aa) edge  [bend left] node [below] {$q-q^2$} (na) 
	(nn) edge [bend left = 15] node [near end, below] {$p^2$} (aa)
	(aa) edge [bend left = 15] node [near end, above] {$(1-q)^2$} (nn)
	(aa) edge  [bend left] node [below] {$q-q^2$} (na);
	\draw [<->]	(na) edge  node [near start, below] {$q-q^2$} (an);
	\end{tikzpicture}
	\caption{Markov chain example depicting attacker behavior (self loops not displayed).}
	\label{fig:aMC1}
	%
	\vspace{-0.35cm}
\end{figure}
The attacker's behavior is represented by the Markov chain in Figure~\ref{fig:aMC1}, where, e.g., the pair $(a,\neg a)$ means that player $0$ is attacked and player $1$ is not attacked. To make sure that the starting probability is $q$, a dummy starting state at time $0$ can be thought as being anything but $(\neg a, \neg a)$. 

\smallskip\par\noindent{\textbf{Defense.}}
Whereas attacks happen at the end of each round, decisions to prepare to defend (denoted $d$), or not to prepare to defend (denoted $\neg d$), are made just before, in the middle of each round. If a player does not defend and the attacker does not attack her, the cost for the player is zero. If the player does not defend and the attacker performs an attack, the cost is $\alpha > 0$. If the player defends, the cost is $\delta > 0$ regardless of the attacker's action.

\smallskip\par\noindent{\textbf{Trading information.}}
Whereas attacks happen at the end of each round, and the decision to defend is made in the middle of each round, trade takes place at the beginning of each round. At the first round, each player $i \in \{0,1\}$ sets a selling price $s^i$ once and for all. At the beginning of every other round, each player may or may not buy, denoted $b$ and $\neg b$, respectively, the piece of information whether the other player was attacked during the previous round. If player $i$ buys and player $1-i$ was attacked, player $i$ pays $s^{1-i}$ to player $1-i$; if player $i$ buys and player $1-i$ was not attacked, player $i$ pays nothing. Moreover, disclosing an attack one  has undergone to the other player costs $s$ to oneself.

\smallskip\par\noindent{\textbf{Playing formally.}}
Let $A := \{a,\neg a\}$, let $B := \{b,\neg b\}$, let $D := \{d,\neg d\}$. A play is a sequence $\sigma d_1 A_1 \cdot b_2 d_2 A_2 \cdot b_3 d_3 A_3 \dots \in (\R^2 D^2 A^2) \cdot (B^2 D^2 A^2)^\mathbb{N}$ where $\sigma = (s^0,s^1)\in \mathbb{R}^2$, $b_n = (b_n^0,b_n^1) \in B^2$, $d_n = (d_n^0,d_n^1) \in D^2$, and $A_n = (A_n^0,A_n^1) \in A^2$ \footnote{The single dot between rounds is just here for readability.}. The $A_n^i$ are capitalized as a reminder that the attacks are random variables, as opposed to the other variables relating to decisions of the players.

For example, the finite prefix $(3,2)(\neg d, d)(a,a) \cdot (b,\neg b)(d,d)(a,\neg a)$ corresponds to two rounds. For the first round, player $0$ and $1$ set selling prices $3$ and $2$, respectively; then only player $1$ defends, but both players are attacked. For the second round, only player $0$ buys information, then both defend, but only player $0$ is attacked.
\smallskip\par\noindent{\textbf{The instant cost for a player during a round.}}
For each player $i \in \{0,1\}$ let $S^i_{s^i}: A^2B^2 \to \mathbb{R}$ describe costs of (not) selling information for price $s^i$, and let $B^i_{s^{1-i}}: A^2B^2 \to \mathbb{R}$ describe costs of (not) buying information for price $s^{1-i}$. In the definition of these functions below, the underscore means any argument.
\begin{align*}
	S^0_{s^0}((a,\_)(\_,b)) & := s - s^0\\
	S^1_{s^1}((\_,a)(b,\_)) & := s - s^1\\
	S^0_{s^0}((\neg a,\_)(\_,\_)) & := S^0_{s^0}((\_,\_)(\_,\neg b)) := S^1_{s^1}((\_,\neg a)(\_,\_))\\& := S^1_{s^1}((\_,\_)(\neg b,\_)) := 0\\
\end{align*}
\begin{align*}
	B^0_{s^1}((\_,a)(b,\_)) & := s^1\\
	B^1_{s^0}((a,\_)(\_,b)) & := s^0\\
	B^1_{s^0}((\neg a,\_)(\_,\_)) & := B^1_{s^0}((\_,\_)(\_,\neg b)) := B^0_{s^1}((\_,\neg a)(\_,\_))\\& := B^0_{s^1}((\_,\_)(\neg b,\_)) := 0
\end{align*}
For each player $i \in \{0,1\}$ let $D^i:D^2A^2 \to \mathbb{R}$ describe the cost of (not) defending.
\begin{align*}
	D^0((d,\_)(\_,\_)) & := D^1((\_,d)(\_,\_)) := \delta\\
	D^0((\neg d,\_)(\neg a,\_)) & := D^1((\_,\neg d)(\_,\neg a)) :=0\\
	D^0((\neg d,\_)(a,\_)) & := D^1((\_,\neg d)(\_,a)) := \alpha
\end{align*}
Let us now define $C^i_{s^0,s^1}$ the total instant cost of a round for player $i$. It is a function from $A^2 \cdot B^2D^2A^2$ to the real numbers. 
Let $C^i_{s^0,s^1}(A_{n-1}\cdot b_nd_nA_n) := S^i_{s^i}(A_{n-1} \cdot b_n) + B^i_{s^{1-i}}(A_{n-1} \cdot b_n) + D^i(d_nA_n)$. Note that $A_{n-1}$ refers to the attacks at the previous round, so the definition does not apply to the first round, which will not be an issue since we are interested in average behaviors.
\smallskip\par\noindent{\textbf{Aggregation of instant costs over time.}}
What each player will optimize is her infinite sequence of expected instant costs, one such a cost per round. There are several classical ways of comparing sequences of real numbers, e.g., mean-payoff, which represents long term objectives; or discounted sum, which represents middle term objectives. Instead, we will use a lexicographic comparison, which represents short term objectives, i.e., each player minimizes her expected instant cost before minimizing the future ones. 




\subsection{Results}\label{sect:game-result}
In this section we show that, under reasonable assumptions, a pattern relating to the attacks triggers the players to buy information. Afterwards, we show that this pattern occurs regularly with probability one. Our special pattern is $A^i_{n} = a \wedge A^i_{n+1} = \neg a$, i.e, player $i$ is attacked at time $n$ but not at time $n+1$. These results are summarized in Theorem~\ref{lem:buy-NE}.

Due to the symmetry between the roles of the two players, several of the concepts and calculations below are only described from the perspective of Player $0$.
\smallskip\par\noindent{\textbf{Useful probability.}}
Let $p' $ be the probability that Player $0$ is attacked at time $n+2$, if she was attacked at time $n$ but not  at time $n+1$. Hence:
\begin{align*}
	p' & = \p(A^0_{n+2} = a \,|\,  A^0_{n} = a \wedge A^0_{n+1} = \neg a)\\
	p' 	& =  \p(A^1_{n+1} = \neg a  \,|\,  A^0_{n} = a) \cdot  \p(A^0_{n+2} = a \,|\, A^0_{n+1} = A^1_{n+1} = \neg a) \\
	&+  \p(A^1_{n+1} = a  \,|\,  A^0_{n} = a) \cdot  \p(A^0_{n+2} = a \,|\, A^1_{n+1} = a)\\
	p' & = (1-q)\cdot p + q \cdot q
\end{align*}
Note that by convex combination we have $p < p' < q$.
\smallskip\par\noindent{\textbf{The cost of not defending.}}
Let us calculate the expected costs of not defending (and not knowing information), after our special pattern.
\begin{itemize}
	\item Let $C(\neg d|(\neg a, \neg a)) = p \alpha$ be the expected cost of not defending at time $n+1$, if no player was attacked at time $n$.
	
	\item Let $C(\neg d|a) = q \alpha$ be the  expected cost of not defending at time $n+1$, if the player was attacked at time $n$.
	
	\item Let $C(\neg d|a \neg a) = p' \alpha$ be the expected cost of Player $0$ not defending at time $n+2$, after the pattern $A^0_{n} = a \wedge A^0_{n+1} = \neg a$.
\end{itemize}
If $\delta \leq p\alpha$, it is optimal for the players to defend all the time, and if  $q\alpha \leq \delta$, it is optimal for the players never to defend. Hence, we assume that $p\alpha < \delta < q \alpha$ in the remainder.
\smallskip\par\noindent{\textbf{The cost of ignorance.}}
Let us now calculate the expected costs of not knowing information, after our special pattern.
\begin{itemize}
	\item Let $C(\neg b | a \neg a)$ be the expected cost of not buying information at time $n+2$, after the pattern $A^i_{n} = a \wedge A^i_{n+1} = \neg a$. So $C(\neg b|a \neg a) = \min (\delta, C(\neg d|a \neg a)) = \min(\delta, p' \alpha)$. Indeed, between defending and not defending, the player will choose one action that minimises the expected cost. 
\end{itemize}
\par\noindent{\textbf{The cost of buying  information.}}
We are now ready to define and calculate $C(b|a \neg a)$, the expected costs of buying information after our special pattern.
\begin{align*}
	C(b|a \neg a)  & =  \p(A^1_{n+1} = \neg a \,|\, A^0_n = a \wedge A^0_{n+1} = \neg a) \cdot p \alpha\\
	& + \p(A^1_{n+1} = a \,|\,A^0_n = a \wedge A^0_{n+1} = \neg a) \cdot (s^1 + \delta)\\
	C(b|a \neg a) & = (1 - q)  p\alpha + q (s^1 + \delta) 
\end{align*}

Lemma~\ref{lem:equi-inequal} below describes two equivalences between inequalities, which are then used to connect formally small selling price and incentive to buy information after our usual pattern (third inequality shows that buying is cost-effective).
\begin{lemma}\label{lem:equi-inequal}
	\begin{enumerate}
		\item\label{lem:equi-inequal1} $(1-q)p\alpha + q(s^1 + \delta) \leq p' \alpha$ iff $s^1 \leq q\alpha - \delta$.
		
		\item\label{lem:equi-inequal2} $(1-q)p\alpha + q(s^1 + \delta) \leq \delta$ iff $s^1 \leq \frac{1-q}{q}(\delta - p\alpha)$.
		
		\item\label{lem:equi-inequal3} $C(b|a \neg a) \leq C(\neg b|a \neg a)$ iff $s^1 \leq \min(q\alpha - \delta, \frac{1-q}{q}(\delta - p\alpha))$.
	\end{enumerate}
\end{lemma}

\par\noindent{\textbf{Average cycling time.}}
Let us now show that our special pattern occurs regularly. Let $L_a := \sup\{l \,|\, A^0_{0} = \dots = A^0_{l-1} = a\}$ be the duration of a non-stop attack against Player $0$ from the beginning. Assuming that $A^0_0 = A^1_0 = 0$, let $L_{\neg a} := \sup\{l \,|\, A^0_{1} = \dots = A^0_{l} = 0\}$ be the duration of a non-stop truce.

Lemma~\ref{lem:length-cycle} below describes precisely the expectation of $L_a$ and gives an upper bound for the expectation of $L_{\neg a}$.
\begin{lemma}\label{lem:length-cycle}
	\begin{enumerate}
		\item $E(L_a ) = \frac{1}{q}$.
		\item $E(L_{\neg a} ) \leq \frac{q(1-p)}{p^2}$
	\end{enumerate}
\end{lemma}
\begin{proof}
	$E(L_a) = \sum_{0}^{+\infty} n \p(L_a = n) = \sum_{0}^{+\infty} n q^n(1-q) = \frac{1}{1-q}$.\\ $E(L_{\neg a}) = \sum_{0}^{+\infty} n \p(L_{\neg a} = n) \leq \sum_{0}^{+\infty} n (1-p)^nq = \frac{q(1-p)}{p^2} + 1$.
\end{proof}
\begin{corollary}
	The expected waiting time between two consecutive special patterns (for a given player) is $E(L_a + L_{\neg a} + 1) \leq 2 + \frac{1}{1-q} + \frac{q(1-p)}{p^2}$.
\end{corollary}
\begin{theorem}\label{lem:buy-NE}
	\begin{enumerate}
		\item If $\delta \leq p\alpha$, it is optimal to defend all the time, so no player ever needs information.
		
		\item If $q\alpha \leq \delta$, it is optimal to never defend, so no player ever needs information.
		
		
		\item If $p \alpha < \delta < q\alpha$, let us further assume that $s \leq \min(q\alpha - \delta, \frac{1-q}{q}(\delta - p\alpha))$. Then players minimizing the costs will choose $s^0 = s^1 = \min(q\alpha - \delta, \frac{1-q}{q}(\delta - p\alpha))$, and information is bought on average at least every $2+ \frac{1}{1-q} + \frac{q(1-p)}{p^2}$ time units.
	\end{enumerate}
\end{theorem}
\begin{proof} We prove the last statement.
		
		The first time that Player $0$ is not attacked is either after the first occurrence of the pattern $A^0_{n-1} = a \wedge A^0_n = \neg a$ or at time $0$, which is as if $A^0_{n-1} = a \wedge A^0_n = \neg a$ had just occurred (since Attacker uses probability $q$ at the beginning.) Since $\delta < q \alpha$, this is first time that Player $0$ may need information. At this time Player $0$ would buy information iff $s^1 \leq \min(q\alpha - \delta, \frac{1-q}{q}(\delta - p\alpha))$. Note that purchase occurs also for $s^1 = \min(q\alpha - \delta, \frac{1-q}{q}(\delta - p\alpha)))$ since the two options give the same expected instant cost, but buying the information may be also useful for future rounds (which is beyond the scope of this paper). Therefore the optimal value of $s^1$ (for Player $1$) is $\min(q\alpha - \delta, \frac{1-q}{q}(\delta - p\alpha))$. Likewise, this is also the optimal value of $s^0$.
	\end{proof}

\subsection{Discussion}\label{sect:model-discussion}

This section discusses how realistic our model is. The conclusion will be that despite its simplicity, our model can incorporate many real-world parameters directly into one of the five abstract parameters, and that slight modifications (and additions) would account for even finer phenomena without altering the overall behavior. Therefore our model provides an accurate insight into the usefulness of our design.

\par\noindent{\textbf{The time frame.}} In the real world, some variables change values quicker than other variables. Moreover, variables of artificial systems are sometimes not allowed to change values quickly to prevent instability. When studying such systems, some variables are thus assumed to be constant over a reasonable time span, which is the case in our model as will be mentioned at several occasions below. Since the premises in Theorem~\ref{lem:buy-NE} are inequalities, as opposed to equalities, we could even replace the ``constant'' assumption by the weaker assumption that the parameters vary within reasonable bounds.

\par\noindent{\textbf{The attacks.}} Correlated attacks: One of the key points of information sharing is that information may help prepare against attacks. This is possible since there are, in the real world, statistical correlations between attacks~\cite{ning2002constructing,allodi2017work}. 

Type of correlation in the real world: correlation between past and present attacks is usually positive. The more attacks occurred recently, the more attacks are likely to occur next. (An extreme but usual case is when attacks occur in bursts.) The correlation between attacks over time on different players is usually also positive. There is a reasonable explanation for that, as attackers often employ a limited arsenal of attacks (e.g., exploits and known vulnerabilities) against a large number of targets~\cite{allodi2017work}, thus leading to the same or similar attacks being observed repeatedly against different targets.
Practical examples could be a malicious IP trying out default password combinations to get access to remote machines of several organizations, or a new strain of malware that is used by more and more attackers.
Our four-state Markov chain is a simple way to express both above-mentioned correlations at once; other more complex expressions would overcomplicate the model, while leading to similar results. After all, we are not interested in modeling all possible
attacks. Our goal is to show that if attacks that match our model take place regularly, it will be cost-effective for players to collaborate regularly as well. Other types of attacks will most likely trigger collaboration even more
often.

Constant Attacker Markov chain: the probabilities $p$ and $q$ are assumed to be (sufficiently) constant over a reasonable time span, but of course they need not be constant in the long run.

Blind or coordinated Attacker: First, the attacker's behavior does not depend on the players' behavior, which is sometimes, but not always, realistic. Indeed, we can say that there are two ends in the attacker behavior spectrum. Either the attacker targets a specific defender with a unique attack (targeted attack), or the attacker targets large groups of defenders with the same attack. In the former case there is little to no correlation to other attacks and therefore collaboration altogether may be futile. In the latter case, the attacker is unlikely to be able to monitor each individual defender closely, and therefore the assumption holds. Second, at any given time, the attacks against each player are drawn independently, which may sound unrealistic at first glance, but it would be possible to show that, in our model, if a player is attacked at a given time, the other player is more likely to be attacked at the same time. This is because the probability which is used by the attacker is either $p$ for both players or $q$ for both players. 

\smallskip\par\noindent{\textbf{The concept of cost aggregation.}} Let us exemplify what aggregation means in this context: someone is cycling to the shop to buy a pair of shoes and a rucksack. The expected cost of this is the price of the shoes, plus the price of the rucksack, plus the probability of the bike being stolen times the price of the bike. One can see here that arbitrarily many (expected) costs can be aggregated into one single (expected) cost without any loss of relevant information, as long as decisions are based on expected costs. This reduces dramatically the number of variables that are required to build a realistic model. The reason why there is one cost variable for defense and one for trading information in our model (rather than just one single variable for both) is because defense and trading do not relate in the same way to the other parts of the model: whether the player defends is her decision (based on available information); whereas trading information depends on the selling price and on the decision of the potential buyer.

\smallskip\par\noindent{\textbf{Defense.}} Independent cost: in our model the cost of defending is the same regardless of the attacker's action. In the real world there may be a difference, but this difference should be small in comparison to the cost difference when not defending (and being attacked or not). For example, assume an organization that finds out about a new attack vector and takes actions to patch it (defend). The cost of defending
(paying developers for patching, system downtime, etc.) is mostly independent from whether or not an attacker will actually try to use this attack vector; especially compared to the cost of falling victim to the attack.
Hence our simplified approximation.

Constant cost: the cost of defending is assumed to be (sufficiently) constant over a reasonable time span, but of course it need not be constant in the long run.

What makes defending costly: First, note that the prepare/not prepare (defend/not defend) modeling is typical of game-theoretic security analysis~\cite{manshaei2013game}. Second, preparing to defend may involve having the company's IT team fully dedicated to this task (thus delaying other tasks), employing extra human personnel (paying extra-hours to employees or temporarily hiring external consultants), limiting access to organizational resources (e.g. during the application of a security patch), reverse engineering malware, etc. Note that the defending cost is incurred per time unit, so it does not take into account fixed costs such as buying (as opposed to renting) special equipment or hiring a security expert for a long time span. However, this is not a drawback of our model
in particular, as it is intrinsic to fixed costs: indeed, it is irrelevant to ``decide'' every time unit if we have bought equipment or hired experts in the past.

\smallskip\par\noindent{\textbf{Trading information.}} Cost of reporting an attack: As mentioned in the beginning of Section~\ref{sec:model}, the cost of sharing information is the aggregation of all the activities related to information sharing. These activities may or may not include writing a report meant to be shared, accessing the sharing platform, taking the risk of information leakage (security, privacy), etc. It is assumed to be (sufficiently) constant over a reasonable time span, but of course it need not be constant in the long run.

Selling price difference: As a matter of design, the price depends on information content. It makes sense for the following reason: reporting the absence of an attack is easy and safe in most cases, whereas reporting an attack in detail costs more, as suggested in the paragraph above. It is possible to prove that, in our model, having a single price would still work, but players would purchase less often.

Constant selling price: In our model the selling prices are chosen by the players once and for all. This only means that the prices are assumed to be stable over a reasonable time span in the real world (to make it possible to build a predictable seller-buyer relationship).
\smallskip\par\noindent{\textbf{Expected costs.}} The costs considered in the model are generally viewed as constant (over a reasonable time span). This may, at first, cause confusion, as e.g. the cost of defending when all that is necessary is adding an IP in a blacklist is
orders of magnitude less than the cost of reverse-engineering new malware, developing and deploying a patch, etc. However, this is not a problem for the model. The players calculate the costs before they actually occur, i.e. they are expected costs, and make decisions based on this expectation.
Effects of the variation become negligible in the long run.

\smallskip\par\noindent{\textbf{Aggregation of instant costs over time.}} Our lexicographic comparison has two advantages, compared to the aforementioned mean-payoff, etc: first, in the current setting it makes calculations and proofs easier; second, in our model, information will be bought mainly if it is worth instantly, whereas in other models it would also be bought if it could be useful later. Yet we will still be able to prove that it is bought regularly, and then to conclude that it would also be bought regularly (even more often) in other models with farther-sighted players.

\smallskip\par\noindent{\textbf{Reputation.}} Although we use neither explicit parameters nor dedicated mechanisms to express or process reputation in the model, our model can still account for fixed reputations over the relevant time spans.

Seller reputation: if we want to model that the buying player does not fully trust the information provided by the selling player, we can modify the corresponding probabilities $p$ and $q$ so as to decrease their difference (only for the other selling player, not for the player herself). Solving such a refined model would be similar but a bit more tedious than what we did in Section~\ref{sect:game-result}.

Buyer reputation: if we want to model that the selling player does not fully trust how the buying player will use the sold information, she can simply incorporate this distrust in the cost of sharing information.

\smallskip\par\noindent{\textbf{Number of players.}} In our model there are only two defending players. If we want to consider a large number of players, the selling price should become near-optimal for every player (for game-theoretic reasons), thus further supporting the exchange of information.
	
In addition, while the assumption $p \alpha < \delta < q\alpha$ seems natural (outside of this interval, there is nothing to do), the condition $s \leq \min(q\alpha - \delta, \frac{1-q}{q}(\delta - p\alpha))$ reminds us that this comes from a specific model. However, the latter condition should be understood as ``if the disclosure cost is small enough'', and the conclusion should then hold in other models.
	
Lastly, in our model, other patterns than our special pattern may trigger purchase of information, possibly for smaller selling prices, but we have not investigated this. Indeed our main point was only to show that information is exchanged on a regular basis if the disclosure cost is small enough.

\section{Prototype and evaluation}\label{sec:prototype}
In this section, we present our prototype implementation of \textit{TRIDEnT}. Specifically,
we present \emph{eth-TRIDEnT}, running on Ethereum. The code can straightforwardly be modified to run on permissioned platforms, e.g., Hyperledger Fabric.
\smallskip\par\noindent\textbf{Smart contract pseudocode:}\label{sec:appendix:pseudo}
In Smart Contract \ref{fig:contract} we provide a pseudocode description of the most notable parts of the backbone smart contract $\mathbb{B}_{\text{market}}$, which provides the basic functionality of \emph{TRIDEnT}. We loosely
follow the notation of~\cite{kosba2016hawk}, with \$ prepended to variables associated to some form of currency (tokens), and \emph{CLT}
standing for the client (e.g. Ethereum wallet address) who calls the function. The initialization of the data structures is omitted for space reasons. The logic of the contract is not trivial, and therefore we
proceed to provide a brief overview.

The \emph{register} function allows clients to register in the system with a public key \texttt{pk} and initializes their balance (ledger) and
ratings set. This initial balance provided in Ether cannot be retrieved from the contract and is therefore considered a form of coin burning even though the value invested is useful inside the system\footnote{Traditional ``wasteful'' burning can naturally also be employed.} . Only registered clients can call any of the following functions. The \emph{advertise} function allows clients to post advertisements
accompanied by tags (see line~\ref{line:adv}) in the marketplace. The function \emph{rmAdvert} allows clients that have posted an advertisement to remove it from the marketplace, while returning fees and deposits held by the contract to their rightful recipients. The \emph{mkOffer} function
adds an offer to a specific advertisement already existing in the marketplace. A fee (initial payment to the seller if the offer is accepted) and a deposit (incentive to provide rating) are withheld by the contract. The \emph{delOffer} function allows the offer maker or the advertiser to
delete (take back or turn down respectively) an offer and send the tokens withheld by the contract to the offer maker. Moving on, the \emph{accOffer} function allows clients that have made advertisements to accept offers for them, claiming the initial fee and providing the encrypted
(ip:port) information needed to establish the stream (see $\mathcal{SE}$ protocol of Figure~\ref{pr:stream}). The \emph{unsubscribe} function allows
clients that have either posted an advertisement or are subscribed to one, to remove their subscription. Finally, \emph{rate} allows clients to
provide a rating about a client that made an advertisement about a stream they are subscribed to. Upon their first rating for a given subscription, they receive their deposit back.
\begin{smartcontract}
	{\footnotesize
	\caption{The Backbone contract $\mathbb{B}_{\text{market}}$}
\label{fig:contract}
\SetKwProg{generate}{Function}{:}{end}
	\generate{register}{
		Upon receiving (\texttt{register}, \texttt{\$payment}, \texttt{pk}) from \emph{CLT}:\\
		ledger['\emph{CLT}'] := \texttt{\$payment}\\
		ratings['\emph{CLT}'] := $\emptyset$\\
		key['\emph{CLT}'] := \texttt{pk}\\
		parties := parties $\cup$ (CLT)\\
	}
	\generate{advertise}{
		Upon receiving (\texttt{advertise}, \texttt{tags}) from \emph{CLT}:\\
		\textbf{ASSERT} \emph{CLT} $\in$ parties\\ \tcc{This assertion is repeated in all subsequent funtions, but omitted for clarity}
		adv = (\texttt{tags}, \emph{CLT}, \{$\emptyset$\}, \{$\emptyset$\})\label{line:adv}\\
		adverts := adverts $\cup$ adv \\
	}
	\generate{rmAdvert}{
		Upon receiving (\texttt{rmAdv}, \texttt{adv}) from \emph{CLT}:\\
		\textbf{ASSERT} \texttt{adv}$\,\in\,$adverts, \emph{CLT} = \texttt{adv}[1] \tcp*{see line~\ref{line:adv}}
		adverts := adverts $\setminus$ \emph{adv}\\
		Return pending offer fees and deposits to offer makers
	}
	\generate{mkOffer}{
		Upon receiving (\texttt{makeOffer}, \texttt{adv}, \texttt{\$fee}, \texttt{\$deposit}) from \emph{CLT}:\\
		\textbf{ASSERT} \texttt{adv} $\in$ adverts\\
		ledger['\emph{CLT}'] := ledger['\emph{CLT}'] - \texttt{\$fee} - \texttt{\$deposit}\\
		ledger[this] := ledger[this] + \texttt{\$fee} + \texttt{\$deposit}\\
		\tcc{tokens transferred to the contract itself (this)}
		offer = (\emph{CLT}, \texttt{fee}) \tcp*{offer[0] = CLT, offer[1] = fee}
		\texttt{adv}[2] := \texttt{adv}[2] $\cup$ offer\\
	}
	\generate{delOffer}{
		Upon receiving (\texttt{dlOffer}, \texttt{offer}) from \emph{CLT}:\\
		\textbf{ASSERT} $\exists\,$adv $\in\,$adverts: offer $\in\,$adv[2], (\emph{CLT} = offer[0] OR \emph{CLT} = adv[1]) \tcp*{CLT is advert owner or has made the offer}
		\texttt{adv}[2] = \texttt{adv}[2] $\setminus$ \texttt{offer}\\
		Return \texttt{\$fee} and \texttt{\$deposit} to offer maker
	}
	\generate{accOffer}{
		Upon receiving (\texttt{acceptOffer}, \texttt{offer}, \texttt{c}) from \emph{CLT}:\\
		\textbf{ASSSERT} $\exists\,$adv $\in\,$adverts: offer $\in\,$adv[2], \emph{CLT} = adv[1]\\
		sub = (\texttt{offer}[0], false, \texttt{offer}[1], \texttt{c}) \tcp*{c = Enc(ip:port)}
		adv[3] = adv[3]  $\cup$ sub\\
		ledger['\emph{CLT}'] := ledger['\emph{CLT}'] + \texttt{offer[1]} \tcp*{claim initial \$fee}
		ledger[this] := ledger[this] - \texttt{offer[1]}\\
		delOffer(\texttt{offer})\\
	}
	\generate{unsubscribe}{
		Upon receiving  (\texttt{unsubscribe}, \texttt{sub}) from \emph{CLT}:\\
		\textbf{ASSERT} $\exists\,$adv $\in\,$adverts: \texttt{sub} $\in\,$adv[3], (\emph{CLT} = adv[1] OR \emph{CLT} = \texttt{sub}[0]) \tcp*{CLT is advert owner or subscriber}
		adv[3] = adv[3] $\setminus$ \texttt{sub}\\
	}
	\generate{rate}{
		Upon receiving (\texttt{rate}, \texttt{sub}, \texttt{rating}) from \emph{CLT}:\\
		\textbf{ASSERT} $\exists\,$adv $\in\,$adverts: \texttt{sub} $\in\,$adv[3]
		\emph{CLT} = \texttt{sub}[0],  \tcp*{CLT is a subscriber}
		ratings[adv[1]] := ratings[adv[1]] $\cup$ \emph{rating}\\
		\tcc{if rating for the first time}
		\uIf{\texttt{sub}[1] = false}{
		ledger['\emph{CLT}'] := ledger['\emph{CLT}'] + deposit\\
		ledger[this] := ledger[this] - deposit\\
	}		sub[1] := true\\
}}
\end{smartcontract}

\smallskip\par\noindent\textbf{eth-TRIDEnT:}
We implemented our smart contracts in Solidity v0.4.25. We created two contracts, namely {\tt{AlertExchange.sol}} and {\tt{Token.sol}}. The former is the heart of the prototype,
implementing the functions described in pseudocode above;
the latter implements a simple ERC-20 token used as the system's currency. In order to mint the token, a node has to burn ETH. Micro-transaction channels were implemented by calls to the already deployed $\mu$-Raiden contract. 
When it comes to evaluating smart contracts, transaction costs (in Ethereum: Gas fees), i.e. fees paid by users to miners, are an important metric and the usual way of arguing about the feasibility of an approach.

In Table~\ref{tab:costs}, we show the transaction costs of the functions of the Alert Exchange contract, as calculated after deployment to the Rinkeby test network. The smart contracts can be found at the following addresses on Rinkeby:
\smallskip{\small \\--\textbf{Alert Exchange:} 0x682528b9cc9b74ca00efb004a4619fabd6f5cd69
\\--\textbf{Token:} 0x380ab69f3230c1990c9bbf4ecdcfefc6bd26501c
\\--\textbf{Raiden Channels:} 0x83aeb45854e1ac54f5d9fa42fd7a79b398aa50cf
\begin{table}[]
\centering
\begin{tabular}{c | r r r r}
\toprule
\multicolumn{1}{c}{\cellcolor{black!30}} & \multicolumn{1}{c}{\cellcolor{black!30}}& \multicolumn{1}{c}{\cellcolor{black!30}\textbf{Cost}} & \multicolumn{1}{c}{\cellcolor{black!30}}&\multicolumn{1}{c}{\cellcolor{black!30}} \\
\multicolumn{1}{c}{\cellcolor{black!30}\textbf{Function}}& \multicolumn{1}{c}{\cellcolor{black!30}\emph{Gas}} & \multicolumn{1}{c}{\cellcolor{black!30}\emph{Gwei}}& \multicolumn{2}{c}{\cellcolor{black!30}\hspace{5pt}\emph{EUR}\hspace{5pt}} \\
\multicolumn{1}{c}{\cellcolor{black!30}}& \multicolumn{1}{c}{\cellcolor{black!30}} & \multicolumn{1}{c}{\cellcolor{black!30}\emph{current}}& \multicolumn{1}{c}{\cellcolor{black!30}\hspace{5pt}\emph{current}\hspace{5pt}} & \multicolumn{1}{c}{\cellcolor{black!30}\hspace{5pt}\emph{peak}} \\ \hline
\multicolumn{1}{c}{\cellcolor{black!30}\tt{deploy}}                    & 3\,994\,723           & 15\,978\,892	        & 2.88			& 99.68		\\ \hline
\multicolumn{1}{c}{\cellcolor{black!30}\tt{register}}                  & 54\,672             & 218\,688        		& 0.04			& 1.36		\\ \hline
\multicolumn{1}{c}{\cellcolor{black!30}\tt{advertise}}                 & 173\,279            & 693\,116             	& 0.12			& 4.32		\\ \hline
\multicolumn{1}{c}{\cellcolor{black!30}\tt{rmAdvert}}                  & 41\,257             & 165\,028        		& 0.03			& 1.03		\\ \hline
\multicolumn{1}{c}{\cellcolor{black!30}\tt{mkOffer}}                   & 194\,381            & 777\,524         	& 0.14			& 4.85		\\ \hline
\multicolumn{1}{c}{\cellcolor{black!30}\tt{delOffer}}                  & 25\,820             & 103\,280        		& 0.02			& 0.64		\\ \hline
\multicolumn{1}{c}{\cellcolor{black!30}\tt{accOffer}}                  & 756\,014            & 3\,024\,056            	& 0.54			& 18.86		\\ \hline
\multicolumn{1}{c}{\cellcolor{black!30}\tt{unsubscribe}}               & 34\,139             & 136\,556    		& 0.02			& 0.85		\\ \hline
\multicolumn{1}{c}{\cellcolor{black!30}\tt{rate}}                      & 46\,663             & 186\,652         	& 0.03			& 1.16		\\ \hline
\bottomrule
\end{tabular}
\smallskip
\caption{Trans. costs at the time of writing (Sep'18) and the peak of Ethereum price and network load (Jan'18) (last col.).}
\label{tab:costs}
\vspace{-0.5cm}
\end{table}}
\smallskip\\The code of the main contract ({\tt{AlertExchange}}) was verified and can be easily inspected at Etherscan\footnote{\url{https://rinkeby.etherscan.io/address/0x682528b9cc9b74ca00efb004a4619fabd6f5cd69}}.
For the calculation of the transaction costs, the gas price was set to 4 Gwei ($4\cdot10^{-9}$ ETH) and the ETH/EUR exchange rate to 180.00, according to observations at the time of writing\footnote{According to \url{https://ethgasstation.info/}. September 2018.}.
Among the transaction costs of the methods, the cost for accepting an offer stands out. This is because the method contains the deletion of an offer, as well as the creation of a subscription.
Subscriptions are the contracts' biggest data structures, which makes them relatively expensive to save.
As an example, the establishment of a stream would cost a total of under 1\,EUR at the time of writing and under 30\,EUR considering the worst-ever price and network traffic\footnote{15th January 2018. ETH/EUR: 1134.20, Average gas price: 22\,Gwei.}. As these streams are projected to be relatively long-lived, these costs are expected to be negligible for firms.
Note that the contract deployment cost is a one-time cost for the whole lifetime of the system.
The implementation of the {\tt{AlertExchange}} contract consists of 423 lines of Solidity code and
a detailed report on the engineering obstacles we encountered will follow in an extended version of the paper.
\smallskip\par\noindent\textbf{Client-side application: }
In order to make interaction with the deployed smart contracts user-friendly, we also
developed a client-side command-line application consisting of around 2\,000 lines of
JavaScript code.
The application calls the smart contract methods using the web3
JavaScript API\footnote{https://github.com/ethereum/wiki/wiki/JavaScript-API}.
It enables all necessary functionalities for viewing, advertising, subscribing,
and rating alert streams, as well as calculating trust values for participants.
Furthermore, the client-application allows parties to provide an endpoint from which
subscribers can download alerts, authenticates the
subscribers and validates that the alert batches are paid for.

\section{Related work}\label{sec:related}
In this section we focus on the related work with regard to incentives and trust for \acp{CIDS}, as well as on collaborative platforms that exist nowadays. For a holistic view on game-theoretical approaches in a variety of network security setting, and economic studies of information sharing, we refer the reader to the surveys of Mansaei et al.~\cite{manshaei2013game}, and Laube and B\"ohme~\cite{laube2017strategic}. For a summary of \acp{CIDS} the reader can refer to Section~\ref{sec:background}.

Gal-Or and Ghose~\cite{gal2005economic} propose a game-theoretic model for information-sharing in cyber-security. They conclude that although information sharing is beneficial to firms, with no additional incentives and anti-free-riding mechanisms, they may not be encouraged to participate truthfully. This work acted as a motivation for \textit{TRIDEnT}. 
Despite the significant work that has been done in the areas of collaborative security \cite{meng2015collaborative} and \acp{CIDS} \cite{vasilomanolakis2015survey,zhou2010survey}, only a very small portion of it deals with incentives and trust among the participants. 
That is, the majority of the related work assumes that all the participants are honest, trustworthy and willing to collaborate.
Duma et al. \cite{Duma2006}, were one of the first to propose a simple trust model for \acp{CIDS}. However, their model is prone to insider attacks (e.g., betrayal attacks).
Fung et al. also examined this field, proposing more advanced mechanisms in their incremental works \cite{fung2008trustids,fung2011dirichlet}, and attempted to deal with the aforementioned insider attack challenge.

Zhu et al. \cite{zhu2012guidex}, touch the topic of incentives in \acp{CIDS} and propose a game-theoretical model that ensures that peers of the system contribute equally (in terms of computational power). This work is very relevant for us, however it is limited only to incentives that are connected to the computational power of the overall system as well as to the resource allocation problem. Hence, this work does not take into account the bigger picture of collaborative security; e.g., the incentives for an organization to join a collaborative ecosystem, the economics of such a system, etc. Similarly, Guo et al. investigate incentives for \acp{CIDS} in mobile ad-hoc networks~\cite{guo2018incentive}. They propose an auction process using virtual credits and model cooperative detection as an evolutionary game.
Although they also use a form of virtual currency, their setting is substantially
different than ours (completely different assumptions) and they do not consider trust relations among collaborators. 
Jin et al.~\cite{jin2017tradeoff} propose a privacy-protection mechanism for \ac{CIDS} where participants
are trustworthy and misreport (share false information) with a known probability to protect their privacy. In our view, misreporting is not an option for real-world systems, and privacy must be preserved by
proper data sanitization and forming collaboration overlays with trustworthy partners, as enabled by \textit{TRIDEnT}. Finally, the idea of using
blockchain technology and smart contracts to enable collaboration in the security domain has been pitched
in several interesting recent works~\cite{alexopoulos2017towards,meng2018intrusion,webstersharing}.
However, to the best of our knowledge, no comprehensive incentive analysis or concrete instantiation for network alert sharing have been proposed before. In the start-up world, \emph{PolySwarm}~\cite{polyswarm} is creating a prediction marketplace for the classification of malware.
This is an interesting idea and could be integrated with \emph{TRIDEnT} on a single platform.

Finally, besides the aforesaid (mostly academic) work, a number of practical real-world alert data sharing platforms exist. The \acf{MISP} is one such example of a collaborative platform that has been used by several organizations \cite{wagner2016misp}. Initially \ac{MISP}, as its original name implies, mainly focused on malware exchange; over the years though, it has been practically extended to more alert types over the years. However, \ac{MISP} does not tackle the incentives topic, nor the trust and alert quality aspect. In addition, due to its structure (closed system that requires a formal process to enroll) attacks on it (internal or external) have not been extensively investigated. 
Other similar systems that have been proposed include IBM's X-Force Exchange\footnote{https://exchange.xforce.ibmcloud.com}, Facebook's ThreatExchange\footnote{https://developers.facebook.com/programs/threatexchange} and DShield\footnote{https://www.dshield.org}. 
These platforms are in principle closed systems that focus on a one-to-many or many-to-many exchange of low level (e.g., raw) data between entities \cite{laube2017strategic}. In many cases, the diversity of shared data is very limited; for instance, the Facebook ThreatExchange is only dealing with Facebook alerts to assist Facebook API developers.
Note that we do not view \emph{TRIDEnT} as a competitor to existing sharing platforms, which have solved a lot of practical problems concerning the structure and mechanics of the actual data transmission and correlation (with the most promising, in our opinion, being \ac{MISP}), but rather as an underlying important incentives layer.
\vspace{-5pt}

\section{Generalization and limitations}\label{sec:limits}
In this section, we discuss how our system can generalize to other types of shared data, as well as what limitations
exist.
\par\noindent\textbf{Generalization}
In essence, \textit{TRIDEnT} is a generic streaming marketplace, and thus could be used for sharing any kind of information,
and not just security alerts. However,
the importance, numerous challenges, as well as the unique risk incurred by both buying and selling alerts, make the functionalities offered by \textit{TRIDEnT} necessary.
Despite its limitations (see below), we consider \textit{TRIDEnT} to be a considerable step towards
increased security collaboration in practice; not only among big companies and organizations, but open to
anyone. Also, the format of \textit{TRIDEnT}'s advertisement and alerts is naturally not limited to Snort-style
alerts. The system can be used to establish streams where e.g. bloom filters or pre-trained machine learning
models are exchanged, allowing greater real-time detection capabilities. The stream establishment functionality
could even be employed in order to run privacy-preserving CIDS on raw data using secure multiparty computation (SMC) in the
future. The main challenges of incentives, alert quality and trust would still be present to some extent in any case.
We consider the field to be important and ripe for research progress after a stagnant period.
\par\noindent\textbf{Limitations}
Considering our system's limitations, an important one is that it has not been tested in the wild, since the financial cost
of developing and supporting production-ready software (and more so security-critical software), greatly exceeds our research budget.
Moreover, the game-theoretic model necessarily abstracts away some variables of the real world and can be extended to take into account
more phenomena (such as synergized returns). Additionally, the stream reselling attack (see Section~\ref{sec:attacks})
may require new mitigation mechanisms.

\section{Conclusion}\label{sec:concl}
A lack of incentives has resulted
in underwhelming adoption of \acp{CIDS}, although technical means are available.
%
%
In this work, we described the end-to-end development of \textit{TRIDEnT}, a decentralized marketplace for
network alert sharing. We started from the inception of the concept, and proceeded to the
system design, game-theoretic analysis, and full prototype implementation.
Our simple, yet insightful game-theoretic analysis of the problem may be of independent
interest, and is open to further development. Although all details and parameters of
a real-world deployment of such a system may be excessively difficult to capture in a
game-theoretic model, our model is a necessary first step and offers a degree of validation to
our concept.
Ideally, our next step would be to find interested parties and experiment with
a real-world \textit{TRIDEnT} alert exchange network, possibly with MISP as the overlay
platform (and TRIDEnT as the trust underlay).

As a future step, it would be interesting to
model interdependencies among the defenders~\cite{kunreuther2003interdependent, laszka2015survey}
in TRIDEnT and craft incentives to increase overall security investment.
Other open challenges include support for varying degrees of data sanitization
depending on the trust relation among collaborators, and
privacy-preserving reputation and overlay formation. Finally, the market should be
analyzed towards its tendencies to build monopolies, as this could be detrimental to its success.
As a concluding remark, we believe that the community should invest more resources in
collaborative security research, especially in light of recent advances in the field, like privacy-preserving machine learning.

\begin{acronym}
	\acro{IDS}{Intrusion Detection System}
	\acro{CIDS}{Collaborative Intrusion Detection System}
	\acro{DL}{Distributed Ledger}
	\acro{DDoS}{Distributed Denial of Service}
	\acro{DoS}{Denial of Service}
	\acro{SPoF}{Single Poing of Failure}
	\acro{IoT}{Internet of Things}
	\acro{EVM}{Ethereum Virtual Machine}
	\acro{SPoF}{Single Point of Failure}
	\acro{MISP}{Malware Information Sharing Platform}
\end{acronym}

\bibliographystyle{alpha}
\bibliography{./bib/more,./bib/Cybersecurity}
\newpage
\appendix
\section{Significance of trust estimates}\label{app:setN}
$N$ is a threshold number of evidence needed to achieve
a desired level of significance of the estimate, in the sense of a confidence interval. The value $N$ is
given by the following expression (from~\cite{hauke2015statistics}), assuming a desired significance parameter $z$, and a certainty level $c>0$,
which is the length of the $100(1-z)\%$ confidence interval.
{\small
	\begin{equation}\label{eq:n}
	N = \frac{-\kappa^2 \cdot (2u^2-4t+4t^2)+ \sqrt{4u^2 \cdot \kappa^4 \cdot (1-u^2) + \kappa^4 \cdot (2u^2-4t+4t^2)^2}}{2u^2}
\end{equation}}
where $u = 1-c$ and $\kappa$ the $100(1-\frac{z}{2})$ percentile of the standard normal distribution. It is notable that the value of N computed in this way is not constant, yet
it depends on the ration of positive and negative experiences considered.

Assume we want to collect evidence until we have enough ($N$) to calculate an estimate with a $100\cdot(1-z)\%$
confidence interval (CI) of length $1-c$. The a priori information, in the form of the baseline trust value,
fades away gradually, and disappears completely when the number of collected evidence reaches $N$.
In our case, we choose to consider an $80\%$ CI of length $0.2$. This means that the actual value associated with a point estimate $t$,
falls with probability $0.8$ in the interval $(t-0.1,t+0.1)$. This leads to a number of required evidence to fade
out the prior value completely, in the ``best-case scenario'' (point estimate - average value $t=0$ or $t=/1$), $N=7$ pieces of
evidence. In the worst case scenario ($t=0.5$), $N=40$. Even in the worst case, 14 evidence pieces are enough to fade out
the initial value almost completely ($c_e=0.92$). Therefore, an implementation can either
compute $N$ from equation~\ref{eq:n} or
choose a constant $N$ value in the range of $[10,15]$, which, by rule of thumb, will
lead to similar results.

\end{document}